\documentclass[twocolumn,10pt]{IEEEtran}

\usepackage{braket}

\renewcommand{\baselinestretch}{0.82}

\input{def.tex}
\usepackage{dsfont}
\usepackage{minibox}
\DeclareSymbolFont{matha}{OML}{txmi}{m}{it}% txfonts
\DeclareMathSymbol{\varv}{\mathord}{matha}{118}
\usepackage{eurosym}
\usepackage{hhline,physics}

\usepackage{multicol}

\makeatletter

\IEEEoverridecommandlockouts

\begin{document}
	\title{{\huge Asymptotic Analysis of Max-Min Weighted SINR for IRS-Assisted MISO Systems with Hardware Impairments}}
	\author{Anastasios Papazafeiropoulos, Cunhua Pan, Ahmet Elbir, Van  Nguyen, Pandelis Kourtessis, Symeon Chatzinotas \thanks{A. Papazafeiropoulos is with the Communications and Intelligent Systems Research Group, University of Hertfordshire, Hatfield AL10 9AB, U. K., and with SnT at the University of Luxembourg, Luxembourg. C. Pan is with the School of Electronic Engineering and Computer Science at Queen Mary University of London, London E1 4NS, U.K. A. M. Elbir is with the Department of Electrical and Electronics Engineering, Koc University, Istanbul, Turkey, and SnT at the University of
			Luxembourg, Luxembourg. P. Kourtessis is with the Communications and Intelligent Systems Research Group, University of Hertfordshire, Hatfield AL10 9AB, U. K. V.-D. Nguyen and S. Chatzinotas are with the SnT at the University of Luxembourg, Luxembourg. E-mail: tapapazaf@gmail.com. This work was supported by Luxembourg National Research Fund (FNR) under the CORE project RISOTTI C20/IS/14773976 and  by the University of Hertfordshire’s 5–year Vice Chancellor’s Research Fellowship.
			
		}}	
	\maketitle

	\begin{abstract}
		We focus on the realistic maximization of the uplink minimum signal-to-interference-plus-noise ratio (SINR) of a general multiple-input single-output (MISO) system assisted by an intelligent reflecting surface (IRS) in the large system limit accounting for HIs. In particular, we introduce the HIs at both the IRS (IRS-HIs) and the transceiver HIs (AT-HIs), usually neglected despite their inevitable impact. Specifically, the deterministic equivalent analysis enables the derivation of the asymptotic weighted maximum-minimum SINR with HIs by jointly optimizing the HIs-aware receiver, the transmit power, and the reflect beamforming matrix (RBM). Notably, we obtain the optimal power allocation and reflect beamforming matrix with low overhead instead of their frequent necessary computation in conventional MIMO systems based on the instantaneous channel information. Monte Carlo simulations verify the analytical results which show the insightful interplay among the key parameters and the degradation of the performance due to HIs.
	\end{abstract}
	
	\section{Introduction}
	Intelligent reflecting surfaces (IRSs), consisted of low-cost, passive reflecting elements with adjustable phase shifts, have been recognized as a promising solution for enhancing the spectral and energy efficiency of wireless systems \cite{Basar2019}. Notably, a significant amount of research has been devoted to IRS-aided  systems \cite{Wu2019a,Elbir2020,Kammoun2020,Papazafeiropoulos2021a,VanChien2021}. For example, in \cite{Wu2019a}, a minimization of the transmit power at the base station (BS) subject to individual signal-to-interference-plus-noise ratio (SINR) constraints took place to address the  transmit and reflect beamforming (RB) optimization problem. Also, in \cite{Kammoun2020}, the optimum linear precoder (OLP) was studied in the large number of antennas regime. 
	
	However, the majority of works in IRS-aided systems have relied on the highly unrealistic assumption of ideal hardware, while practical implementations of IRSs require taking into account the unavoidable residual transceiver hardware impairments (T-HIs) 
	%	. Despite any compensation algorithms, HIs remain \cite{Studer2010}, degrade the performance, and their 
	whose omission may result in misleading design conclusions. Especially, a cost-attractive implementation of massive multiple-input multiple-output (mMIMO) systems suggests the use of cheap hardware, resulting in more severe HIs \cite{Bjoernson2014,Papazafeiropoulos2017a,Papazafeiropoulos2017b}. One major category of T-HIs, known as additive transceiver HIs (AT-HIs), can be modeled as additive Gaussian distributed noise by accounting for the accumulated effect from all individual HIs such as the in-phase/quadrature-phase imbalance  \cite{Studer2010,Bjoernson2014,Papazafeiropoulos2017a,Papazafeiropoulos2017b}. Another interesting type of HIs, called IRS-HIs, emerges in IRS-assisted systems because of the incapability of infinite precision at the IRS phase shifts \cite{Badiu2019,Li2020,Papazafeiropoulos2021}. Hence, the performance analysis of the IRS-aided systems should include the impact of both AT-HIs and IRS-HIs. 
	
	Recently, the study of HIs on IRSs has attracted significant interest \cite{Xing2020,Liu2020, Zhou2020,Shen2021,Zhou2020a,Papazafeiropoulos2021}. The authors in \cite{Xing2020} focused on the achievable rate  while they considered only single‐input and single‐output (SISO) channels. In \cite{Liu2020}, upper bounds on the channel capacities were obtained while relying on the assumption of no correlation among the columns of the IRS, assuming a single user communication, and not providing closed forms regarding these bounds. In \cite{Zhou2020}, no small-scale fading was assumed and just a single user equipment (UE) was considered. The latter limiting design setting was also assumed in \cite{Shen2021} to maximize the received signal-to-noise ratio (SNR)   while, in \cite{Zhou2020a}, the authors focused on the maximization of the secrecy rate for a finite number of BS antennas. In this direction, in \cite{Papazafeiropoulos2021}, we studied the impact of HIs on the achievable rate in a multi-user setting for a finite number of BS antennas without optimizing the transmit power.
	
	In this paper, we make a substantial leap beyond previous works by accounting for both AT-HIs and IRS-HIs in terms of the deterministic equivalent (DE) analysis by formulating a max-min weighted SINR problem in the large antenna regime. \footnote{There are many differences between our work and existing works, e.g., \cite{Shen2021}. Therein, first,  just the SNR was studied for a single UE scenario while we focus on the SINR accounting for multiple UEs.  Second, we consider the max-min optimization problem, while \cite{Shen2021} considered the rate maximization problem. Third, \cite{Shen2021} considered the phase shift design based on instantaneous CSI, while our work designs the phase shift based on the statistical CSI. Fourth, we consider the asymptotic case when the number of transmit antennas is infinite, while \cite{Shen2021} is suitable for a limited number of antennas due to its high complexity when the number of transmit antennas is large.}Contrary to \cite{Kammoun2020}, we introduce both AT-HIs and IRS-HIs, and focus on the uplink instead of the downlink. Therein, an already obtained OLP, based on \cite{Cai2011}, was applied which limits the analysis while our methodology, taking HIs into consideration, is more general in terms of DEs and optimization as the following analysis reveals, which adds to the novelty of this work. For example, we obtain the optimal decoder and the optimal power allocation with HIs by following another approach than \cite{Cai2011} that was based on uplink-downlink duality. Also, we have taken into account the direct channel while its manipulation was not possible by the analysis in \cite{Kammoun2020}. It should be noted that the introduction of HIs, increasing the complexity/difficulty, requires substantial manipulations.
	%		Note that a common user correlation matrix was assumed in \cite{Kammoun2020}, while the proposed work concerns the realistic setting of different correlation matrices across UEs.
	%		{For example, we obtain 
	%%			Also, this work is one of the few that have studied the uplink of IRS-based systems. based on uplink-downlink duality
	%On top of this, we have focused on mMIMO systems and on the impact of HIs.
	% Compared to \cite{Liu2020}, which relied on upper bounds, no correlated fading, and a finite number of BS antennas, we focus on the study of a lower bound by taking correlation into account, and we provide a large system analysis.
	Specifically, by considering the general realistic scenario of correlated Rayleigh fading channels with HIs, we obtain the optimal HIs-aware linear minimum mean square error (LMMSE) receiver and the corresponding asymptotic optimal weighted SINR. We achieve optimal power allocation and an RB design that require only large-scale channel statistics and do not depend on small-scale fading changing at the order of milliseconds that would result in prohibitively high overhead. 
%	Apart from this indirect property by DEs in IRS-assisted mMIMO systems, 
	%Notably, in this way, we achieve the optimization of the IRS phases with only large-scale statistics. 
	The results allow shedding light on the impact of HIs on such systems towards their realistic evaluation.
	%	\footnote{ {Despite their importance, the study of HIs at the IRS is left for future work due to limited space}.}.

	%	\textit{Notation:} Vectors and matrices are denoted by boldface lower and upper case symbols. The symbols $(\cdot)^\T$ and $(\cdot)^\H$ express the transpose and Hermitian transpose operators, respectively. The expectation operator is denoted by $\EE\left[\cdot\right]$ and $ \bONE $ is the vector of all ones. The notations $\mathbb{C}^{M \times 1}$ and $\mathbb{C}^{M\times N}$ refer to complex $M$-dimensional vectors and $M\times N$ matrices, respectively. Also, the notation $\xrightarrow[ M \rightarrow \infty]{\mbox{a.s.}}$ denotes almost sure convergence as $ M\! \rightarrow \!\infty $ and $ \circ $ denotes the Hadamard product. Finally, $\bb \sim \cC\cN{(\b0,\mathbf{\Sigma})}$ denotes a circularly symmetric complex Gaussian variable with covariance matrix $\mathbf{\Sigma}$.
	\section{System Model}\label{System}
	We consider an IRS-aided multi-user massive MIMO system, where a multi-antenna BS with $ M $ antennas communicates with $K $ single-antenna UEs. To focus on the impact of hardware distortions, we assume perfect channel state information (CSI), which allows more direct mathematical manipulations. Hence, the results play the role of upper bounds of practical scenarios with imperfect CSI. Note that the CSI could be assumed perfectly known when the coherence intervals are sufficiently long. The extension to the imperfect CSI scenario, which is of practical importance, is the topic of future work. In particular, the IRS is assumed in the line-of-sight (LoS) of the BS and includes $ N $ passive reflecting elements introducing shifts on the phases of the impinging waves. Also, the IRS is controlled by the BS by means of a backhaul link.
	
	We rely on a block-fading channel model with fixed channels in each time-frequency coherence block but with independent realizations in each block. Specifically, $ \bH_{1}=[\bh_{1,1}\ldots,\bh_{1,N} ] \in \mathbb{C}^{M \times N}$, $ \bh_{\mathrm{d},k} \in \mathbb{C}^{M \times 1} $, and $ \bh_{2,k} \in \mathbb{C}^{N \times 1}$ express the LoS channel between the BS and IRS, the direct channel between the BS and UE $ k $, and the channel between the IRS and UE $ k $. The vector $ \bh_{1,i} $ for $ i=1,\ldots,N $ corresponds to the $ i $th column vector of $ \bH_{1}$. Notably, we consider spatial correlation instead of independent Rayleigh fading assumed in the majority of previous works, e.g., \cite{Wu2019a}. Hence, we have
	\begin{align}\bh_{\mathrm{d},k}&=\sqrt{\beta_{\mathrm{d},k}}\bR_{\mathrm{BS},k}^{1/2}\bz_{\mathrm{d},k},\\
		\bh_{2,k}&=\sqrt{\beta_{2,k}}\bR_{\mathrm{IRS},k}^{1/2}\bz_{2,k},
	\end{align}
	where $ \bR_{\mathrm{BS},k} \in \mathbb{C}^{M \times M} $ with $ \tr\left(\bR_{\mathrm{BS},k} \right)=M $ and $ \bR_{\mathrm{IRS},k} \in \mathbb{C}^{N \times N} $ with $ \tr\left(\bR_{\mathrm{IRS},k} \right)=N $ express the deterministic Hermitian-symmetric positive semi-definite correlation
	matrices at the BS and the IRS, respectively. Also, $ \beta_{\mathrm{d},k} $ and $ \beta_{2,k} $ are the path-losses of the BS-UE $ k $ and IRS-UE $ k $ links. Note that the correlation matrices and the path-losses are assumed to be known by practical methods, {e.g., see} \cite{Neumann2018}. In addition, $ \bz_{\mathrm{d},k}\sim \mathcal{CN}\left(\b0,\Id_{M}\right) $ and $ \bz_{2,k}\sim \mathcal{CN}\left(\b0,\Id_{N}\right) $ express the respective fast-fading vectors. Moreover, we assume that $ \bH_{1} $ is a full rank channel matrix described as 
	\begin{align}
		[\bH_{1}]_{m,n}&=\sqrt{\beta_{1}} \exp \Big(j \frac{2 \pi }{\lambda}\left(m-1\right)d_{\mathrm{BS}}\sin \theta_{1,n}\sin \psi_{1,n}\nn\\
		&+\left(n-1\right)d_{\mathrm{IRS}}\sin \theta_{2,m}\sin \psi_{2,m}\Big),
	\end{align}
	where $ \lambda $ and $ \beta_{1} $ are the carrier wavelength and the path-loss between the BS and IRS, while $ d_{\mathrm{BS}} $ and $ d_{\mathrm{IRS}} $ are the inter-antenna separation at the BS and inter-element separation at the IRS, respectively \cite{Kammoun2020}. In addition, $ \theta_{1,n} $, $ \psi_{1,n} $ express the elevation and azimuth LoS angles of departure (AoD), respectively at the BS with respect to IRS element $ n $, while $ \theta_{2,n} $ and $ \psi_{2,n} $ are the elevation and azimuth LoS angles of arrival (AoA) at the IRS. The design of $ \bH_{1} $ could be realized by several techniques as suggested in \cite{Bohagen2007}. In addition, the response of the IRS elements is described by the diagonal matrix $ \bPhi=\mathrm{diag}\left(\al \exp\left(j \phi_{1}\right), \ldots, \al \exp\left(j \phi_{N}\right)\right)\in \mathbb{C}^{N\times N}$, where $ \phi_{n} \in \left[ 0, 2 \pi \right], n=1,\ldots,N$ and $ \al \in (0,1]$ express the phase shifts applied by the IRS elements and the independent amplitude reflection coefficient, respectively.\footnote{Recently, it was shown that the   amplitude and phase responses are intertwined \cite{Gradoni2021,Zhang2021}, which suggests an interesting idea for extension of the current work, i.e., to study  the impact of active (additive transceiver distortion) and passive (IRS phase noise) HIs by accounting for this intertwinement.}
	\subsubsection{IRS-HIs}	Since it is not possible to configure the IRS elements with infinite precision, phase errors are introduced \cite{Badiu2019}. These IRS-HIs can be described by means of a random diagonal phase error matrix consisting of $ N $ random phase errors, i.e., $ \widetilde{\bPhi} =\diag\left( e^{j \tilde{\phi}_{1}}, \ldots, e^{j \tilde{\phi}_{N}} \right)\in\mathbb{C}^{N\times N}$ with $ \tilde{\phi}_{i}, i=1,\ldots,N $ being the random phase errors of the IRS phase shifts that are i.i.d. randomly distributed in $ [-\pi, \pi) $ and based on a certain circular distribution. \footnote{The probability density function (PDF) of $\tilde{ \theta}_{i} $ is assumed symmetric with its mean direction equal is zero, i.e., $ \arg\left(\EE[\mathrm{e}^{j \tilde{\theta}_{i}}]\right)=0 $ \cite{Badiu2019}. 	} Hence, the channel vector between the BS and UE $ k $ is 	written as $ \bh_{k}=\bh_{\mathrm{d},k}+\bH_{1}\bPhi\tilde{\bPhi}\bh_{2,k}, \in \mathbb{C}^{M\times 1}$, distributed as $ \cC\cN\left( 0, \bR_{k} \right) $, where $ \bR_{k}=\beta_{\mathrm{d},k}\bR_{\mathrm{BS},k}+ \beta_{2,k}\bH_{1} \bPhi\tilde{\bR}_{\mathrm{IRS},k}\bPhi^{\H}\bH_{1}^{\H}$ with $ \tilde{\bR}_{\mathrm{IRS},k}= m^{2}\bR_{\mathrm{IRS},k}+\left(1-m^{2}\right)\Id_{N}$ and $ m $ denoting its characteristic function (CF) \cite[Eq. 12]{Papazafeiropoulos2021}.

	Examples of PDFs that could describe the phase noise on IRSs are the uniform and the Von Mises distributions \cite{Badiu2019}. The former expresses completely lack of knowledge and has CF equal to $0 $ while the latter has a zero-mean and concentration parameter $ \kappa_{\tilde{\theta}} $, 	capturing the accuracy of the estimation. Its CF is $ m= \frac{\mathrm{I}_{1}\!\left(\kappa_{\tilde{\theta}}\right)}{\mathrm{I}_{0}\!\left(\kappa_{\tilde{\theta}}\right)}$, where $ \mathrm{I}_{p}\!\left(\kappa_{\tilde{\theta}}\right)$ is the modified Bessel function of the first kind and order 	$ p $. 
	\begin{remark}\label{rem1}
		If $ m=0 $ (uniform distribution), we obtain $ \widetilde{\bR}_{\mathrm{IRS},k}=\Id_{N}$, which means that $ \bR_{k} $ does not depend on the phase shifts and the system cannot be optimized due to the IRS. Especially, we have $ \bR_{k} =\beta_{\mathrm{d},k}\bR_{\mathrm{BS},k}+\beta_{2,k}\bH_{1}\bH_{1}^{\H}$. In this case, no knowledge of $ \bR_{\mathrm{IRS},k} $ is required at the BS. This result is obtained also if ${\bR}_{\mathrm{IRS},k}=\Id_{N} $, which means that in the case of no IRS correlation, the IRS cannot be optimized if statistical CSI is considered. However, if the phase errors follow any other circular PDF, $ \bR_{k} $ is phase-dependent and the presence of IRS can be exploited.
	\end{remark}
	\subsubsection{AT-HIs} Disregarding most works in the IRS literature assuming ideal transceiver hardware, in practice, HIs, remain despite the use of any mitigation algorithms and affect both the transmit and receive signals. In this direction, we account for additive HIs at both the transmitter and the receiver (AT-HIs) being Gaussian distributed with average powers proportional to the average transmit and receive signals, respectively \cite{Studer2010}. The Gaussianity is a result of the aggregate contribution of many impairments. Notably, this model is not only analytically tractable but it is also experimentally
	validated \cite{Studer2010}. For instance, during the uplink, let $ \frac{p_{k}}{M}=\EE \{|x_{k}|^{2}\}$ be the transmit power from UE $ k $ transmitting signal $ x_{k} $. Then, the AT-HIs are described in terms of conditional Gaussian distributions as 
	%Disregarding most works in the IRS literature assuming ideal transceiver hardware, in practice, HIs, remaining despite the use of any mitigation algorithms, affect the transmit and receive signals. Especially, the implementation of massive MIMO systems is suggested by cheap hardware, resulting in more severe HIs, in order to be a cost-attractive technology. In this direction, we account for additive HIs at both the transmitter and the receiver being Gaussian distributed with average powers proportional to the average transmit and received signals, respectively \cite{Schenk2008,Studer2010}. Given that we focus on the SE performance, we employ Gaussian
	%codebooks maximizing the differential entropy. Hence, during the uplink, let $ \bx =[x_{1}p_{k}=\EE \{|x_{k}|^{2}\}$ be the transmit power from UE $ k $ transmitting signal $ x_{k} $, 
	\begin{align}
		\delta_{\mathrm{t},k}&\sim \cC\cN\left( 0, \Lambda_{k} \right)\!\label{eta_tU}\! \\
		\deltav_{\mathrm{r}}&\sim \cC\cN \left( \b0,\bm \Upsilon \right)\!\label{eta_rU},
	\end{align}
	where $ \Lambda_{k}= \kappa_{\mathrm{UE}}\frac{p_{k}}{M}$ and 
	$\bm \Upsilon =\kappa_{\mathrm{BS}}\sum_{i=1}^{K}\frac{p_{i}}{M} \mathrm{diag}( |h_{i,1}|^{2},\ldots,$ $|h_{i,M}|^{2} ) $ with $\kappa_{\mathrm{UE}}$ and $\kappa_{\mathrm{BS}}$ expressing the severity of the residual impairments at the transmitter and receiver side, respectively. For simplicity, all UEs are assumed with identical HIs, i.e., $\kappa_{\mathrm{UE}_{i}}=\kappa_{\mathrm{UE}}~\forall i$. The extension to different distortion at each UE is straightforward. Note that in the case $ \kappa_{\mathrm{UE}}=\kappa_{\mathrm{BS}}=0 $ and $ m=1 $, we result in the ideal scenario with no HIs.

	\section{Uplink Data Transmission with HIs}\label{PerformanceAnalysis}
	The received complex baseband signal by the BS is
	\begin{align}
		\by=\sum_{i=1}^{K}\bh_{i}\left(x_{i}+\delta_{\mathrm{t},i}\right) +\deltav_{\mathrm{r}}+\bw,\label{ULTrans}
	\end{align}
	where $ \delta_{\mathrm{t},i} $ and $\deltav_{\mathrm{r}} $ are the transmit and receive distortions given by \eqref{eta_tU} and \eqref{eta_rU}, respectively. Also, $\bw\sim \mathcal{CN}\left(\b0,\sigma^2\Id_{M}\right) $ is the receiver noise. The signal of UE $ k $, detected by the combining vector $ \bv_{k}\in \mathbb{C}^{M \times 1} $, can be expressed as $ 	\bv_{k}^{\H}\by $.
	%	\begin{align}
	%	\bv_{k}^{\H}\by=\bv_{k}^{\H}\sum_{i=1}^{K}\bh_{i}\left(s_{i}+\delta_{\mathrm{t},i}\right) +\bv_{k}^{\H}\deltav_{\mathrm{r}}+\bv_{k}^{\H}\bw.\label{ULTrans1}
	%	\end{align}
	\begin{lemma}
		The instantaneous uplink SINR of an IRS-assisted MIMO system with AT-HIs and IRS-HIs is given by
		\begin{align}
			\gamma_{k}=\frac{\frac{p_{k}}{M} \bv_{k}^{\H}\bh_{k}\bh_{k}^{\H} \bv_{k}}{ \bv_{k}^{\H}\!\left(\sum_{i\ne k}\frac{p_{i}}{M}\bh_{i}\bh_{i}^{\H} +\bC_{\!\delta_{\mathrm{t}}} +\bC_{\!\delta_{\mathrm{r}}}+\sigma^2 \Id_{M}\!\right)\! \bv_{k}},\label{SINR}
		\end{align}
		where $ \bC_{\!\delta_{\mathrm{t}}} = \kappa_{\mathrm{UE}}\sum_{i=1}^{K}\frac{p_{i}}{M} \bh_{i}\bh_{i}^{\H}$ and $ \bC_{\!\delta_{\mathrm{r}}}=\kappa_{\mathrm{BS}}\sum_{i=1}^{K}\frac{p_{i}}{M} \Id_{M}\!\circ\!\bh_{i}\bh_{i}^{\H}$.
	\end{lemma}
	\begin{proof}
		Given that the AT-HIs are Gaussian distributed and uncorrelated with the transmit signals, we make use of the worst-case uncorrelated
		additive noise theorem in \cite{Bjoernson2017} to obtain a lower bound of the mutual information $ \mathcal{I} $ between the input $ x_{k} $ and output $ \bv_{k}^{\H}\by $ for a given channel realization $ \bH=[\bh_{1}, \ldots, \bh_{K}]\in \mathbb{C}^{M \times K } $ as $ \EE_{\bH}\{\mathcal{I}\left(x_{k};\bv_{k}^{\H}\by\right)\} \ge\EE_{\bH}\{ \log_{2}\left(1+\gamma_{k}\right)\} $, where $ \EE_{\bH}\{\mathcal{I}\left(x_{k};\bv_{k}^{\H}\by\right)\} $ expresses the ergodic achievable SE with $ \EE_{\bH}\{\cdot\} $ denoting the expectation with respect to $ \bH $, and $ \gamma_{k} $ is given by \eqref{SINR} with $ \circ $ in $ \bC_{\!\delta_{\mathrm{r}}} $ denoting the Hadamard product. 
	\end{proof}
	
	\subsection{Problem Formulation}
	The focal point of this work is the max-min weighted uplink SINR under a weighted sum-power constraint.
	%respect to users’ quality-of-service (QoS) constraints.
	The optimization problem is described as
	\begin{align}
		\!\!	(\mathcal{P}1)~&\max_{ \bV, \bp, \bPhi} \min_{k} \frac{\gamma_{k}\left(\bV, \bp, \bPhi\right)}{\eta_{k}}\label{Maximization1} \\
		%\sum_{i=1}^{K}\log_{2}\left(1+\gamma_{i}\right)\\
		&~\mathrm{s.t.}~~~~~~~~\frac{1}{M}\betav^{\T}\bp \le p_{\mathrm{max}},~p_{k}>0,~ \|\bv_{k}\|=1, \forall k\label{Maximization3} \\
		%\mathrm{s.t.}~\gamma_{i}&=\gamma_{i}^{\min},~~\forall i \in \{1, \ldots,K\}\\
		%p_{i}&\le P_{i}^{\max},~~\forall i \in \{1, \ldots,K\}\\
		&~~~~~~~~~~~~~|\phi_{i}|=1, ~~\forall i \in \{1, \ldots,N\},\label{Maximization4} 
	\end{align}
	where $\bV=[\bv_{1}, \ldots, \bv_{K}]$ and $ \bp=[p_{1}, \ldots, p_{K}]^{\T}$ are the tuple of receive beamforming matrices and the transmit power vector, respectively. Also, $ p_{\mathrm{max}} $ denotes the given power constraint while $ \etav=[\eta_{1}, \ldots, \eta_{K}]^{\T}$ and $\betav=[\beta_{1}, \ldots, \beta_{K}]^{\T}$ with $ \eta_{k} $ and $ \beta_{k} $ expressing the priority assigned to UE $ k $ and the weight associated with $ p_{k} $. 
	
	At optimality, the power coefficients are obtained based on the property that the weighted SINRs for different UEs are identical, i.e., $ \frac{\gamma_{1}}{\eta_{1}}=\ldots= \frac{\gamma_{K}}{\eta_{K}}=\tau^{\star} $\cite{Cai2011}. As a result, the SINR constraint for UE $ k $ is written as
	\begin{align}
		\gamma_{k}\left(\bV, \bp, \bPhi\right)\ge \eta_{k} \tau^{\star} ~~~~~\forall k \label{SINRconstraint}
	\end{align}

	\section{Proposed Design}
	The problem in $ (\mathcal{P}1) $ is non-convex and the coupling among the optimization variables (the active and passive beamforming at the BS and IRS, respectively, and the power control) raises difficulties to solve. We tackle them by following the common  alternating optimization in two stages. The notable difference here is that we take HIs and correlated fading into account. In particular, first, for any given $ \bPhi $, we provide the optimal linear receiver design in terms of the optimal decoders and the optimal allocated power. Then, we consider the IRS design.
	% However, notably, in this work, we rely on the maximization of the rate as a first thing, and then we present the remaining design based on the optimal DE rate.
	
	%For any given $ \bPhi $, a simpler optimization problem for \eqref{Maximization1} can be formulated by only optimizing the power vector $ \bp $ with given the decoder matrix $ \bV $.
	\begin{proposition}
		Given the  RBM $ \bPhi $ and the power vector $ \bp $, the uplink SINR of an IRS-assisted MIMO system with  AT-HIs and IRS-HIs is maximized by the HIs-aware LMMSE receiver
		\begin{align}
			\bv_{k}^{\star}= \!{\bigg(\!\displaystyle\sum_{i=1}^{K}\frac{p_{i}}{M}\bh_{i}\bh_{i}^{\H}\!+\!\bC_{\!\delta_{\mathrm{t}}} \!+\!\bC_{\!\delta_{\mathrm{r}}}\! +\!\sigma^2\Id_{M} \!\!\bigg)^{\!\!\!-1} \bh_{k}},\label{OptimalMMSE2}
			%			\bv_{k}^{*}= \!\frac{\bigg(\!\displaystyle\sum_{i=1}^{K}\frac{p_{i}}{M}\bh_{i}\bh_{i}^{\H}\!+\!\bC_{\!\delta_{\mathrm{r}}}\! +\!\sigma\Id_{M} \!\!\bigg)^{\!\!\!-1} \bh_{k}}{\bigg\|\bigg(\!\displaystyle\sum_{i=1}^{K}\frac{p_{i}}{M}\bh_{i}\bh_{i}^{\H}\!+\!\bC_{\!\delta_{\mathrm{r}}}\! +\!\sigma\Id_{M} \!\!\bigg)^{\!\!\!-1} \bh_{k}\bigg\|},\label{OptimalMMSE2}	
		\end{align}
		and the optimal SINR is obtained as
		\begin{align}
			\gamma_{k}=\frac{p_{k}}{M}\bh_{k}^{\H} \bSigma\bh_{k},\label{SINROptimal}
		\end{align}
		where $\bSigma\!=\!\bigg(\displaystyle\sum_{i\ne k}^{K}\frac{p_{i}}{M}\bh_{i}\bh_{i}^{\H}+ \bC_{\!\delta_{\mathrm{t}}} + \bC_{\!\delta_{\mathrm{r}}} +\sigma^2\Id_{M}\!\! \bigg)^{\!\!-1} $.
	\end{proposition}
	\begin{proof} 
		The SINR $ \gamma_{k} $ in \eqref{SINR} can be written as a generalized Rayleigh quotient that can be maximized according to \cite[Lem. B.10] {Bjoernson2017} by $ \bv_{k}^{\star}$ given by \eqref{OptimalMMSE2}, where the matrix inversion lemma has been also applied. A straightforward substitution of \eqref{OptimalMMSE2} into \eqref{SINR} results in \eqref{SINROptimal}.
	\end{proof}
		%	\begin{remark}
	%		The receive combining vector maximizing vector at the BS does not depend on the UE additive distortion $ \kappa_{\mathrm{UE}} $ but only on the BS hardware quality $ \kappa_{\mathrm{BS}} $. {In \cite{Liu2020}, the receiver has not been simplified and still includes the UE additive distortion.}
	%	\end{remark}
	The optimization of the power allocation relies on a large system analysis of the max-min weighted SINR.\footnote{
		%\section{Large System Analysis}\label{Deterministic}
		The application of MMSE-type receivers to IRS-assisted mMIMO systems includes prohibitively demanding computations such as the matrix inversion as $ M,N,K $ increase. In addition, the corresponding calculations should take place in every coherence interval. These reasons indicate the use of the theory of DEs concerning derivations in the asymptotic limit $M,N,K \to \infty$ while their ratios are kept fixed. Notably, the DE results are tight approximations even for a medium system size, e.g., $ 8 \times 8 $~\cite{Hoydis2013}.} In particular, taking into account  the same  assumptions  concerning the correlation matrices as in \cite[Assump. A1-A2]{Hoydis2013}, we obtain the DE SINR 
	%$ \bar{\gamma}_{k} $ according to $ \gamma_{k}-\bar{\gamma}_{k}\xrightarrow[M \rightarrow \infty]{\mbox{a.s.}}0$.
	$ \bar{\gamma}_{k} $ according to $ \gamma_{k}-\bar{\gamma}_{k}\xrightarrow[M \rightarrow \infty]{\mbox{a.s.}}0$. The notation $\xrightarrow[ M \rightarrow \infty]{\mbox{a.s.}}$ denotes almost sure convergence as $ M\! \rightarrow \!\infty $. 	Hence, the DE weighted SINR is given by $ 	\bar{\gamma}_{k}/\eta_{k}=\bar{\tau}$.
	\begin{lemma}\label{theorem:ULDEMMSE}
		The DE of the optimal SINR of an IRS-assisted mMIMO system, accounting for HIs, is given by
		\begin{align}
			\bar{\gamma}_{k}=p_{k}\frac{\delta_{k}}{1+\kappa_{\mathrm{UE}}p_{k}\delta_{k}},\label{SINRDE1}
		\end{align}
		where $ \delta_{k}=\frac{1}{M}\tr\left(\bR_{k}\bT \right) $ with 
		$ \bT=\!\!\sum_{i\ne k}^{K}\!\frac{p_{i}}{M}\big(\!\frac{\left(1\!+\! \kappa_{\mathrm{UE}}\right)}{\left(1+\delta_{i}\right)}\bR_{i}\!+\!\sum_{i=1}^{K}\!\frac{p_{i}}{M}\kappa_{\mathrm{BS}}\Id_{M}\!\circ\!\bR_{i}\!+\!\sigma^2\Id_{M}\!\big)^{\!\!-1} $.
	\end{lemma}
	\proof
	%	 Refer to Appendix~\ref{theorem1}.
	We define the matrix 
	$ \bSigma_{k}\!=\!\big(\!\sum_{i\ne k}^{K}\!\frac{p_{i}}{M}\left(1+\kappa_{\mathrm{UE}}\right)\bh_{i}\bh_{i}^{\H}\!+\! \bC_{\!\delta_{\mathrm{r}}}\!+\! \sigma^2 \Id_{M}\! \big)^{\!-\!1} $.
	%			where 		$\al$ is a regularization scaled by $M$ to make expressions converge to a constant, as $M$, $K\to \infty$. 
	From \eqref{SINROptimal}, the optimal SINR becomes
	\begin{align}
		\gamma_{k}&\xrightarrow[ M \rightarrow \infty]{\mbox{a.s.}} p_{k}\frac{\frac{1}{M}\tr\left(\bSigma_{k} \bR_{k} \right)}{1+\frac{1}{M}\kappa_{\mathrm{UE}}p_{k}\tr\left(\bSigma_{k}\bR_{k} \right)}\label{SINR2}\\
		&\xrightarrow[ M \rightarrow \infty]{\mbox{a.s.}}p_{k}\frac{\frac{1}{M}\tr\left(\bR_{k} \bT \right)}{1+\frac{1}{M}\kappa_{\mathrm{UE}}p_{k}\tr\left(\bR_{k}\bT \right)},
	\end{align}
	where in \eqref{SINR2}, we have applied the matrix inversion lemma \cite[Lem. 1]{Hoydis2013}, and have used \cite[Lem. 4]{Hoydis2013}.
	%		 and the perturbation lemma (\cite[Lem. 3]{Hoydis2013})
	\footnote{We assume that the diagonal matrix inside $ \bSigma_{k} $ is considered deterministic with diagonal elements given by the limits of the individual diagonal elements~\cite{Papazafeiropoulos2017a}. Specifically, we exploit the uniform convergence $ \lim \sup_M \max_{1\le i\le M} { \left|\left[\bh_{i}\bh_{i}^{\H} \right]_{mm} - \left[\bR_{i}\right]_{mm}\right| } = 0$, and obtain $\left\|\frac{1}{M} \mathrm{diag}(\bh_{i} \bh_{i}^{H})-\frac{1}{M}\tr\left(\diag\left(\bR_{i}\right)\right)\!\right\| \xrightarrow[ M \rightarrow \infty]{\mbox{a.s.}}\! 0 $ where $ [A]_{mm} $ denotes the $ m $th diagonal element of matrix $ \bA $.} The last step makes use of \cite[Th. 1]{Hoydis2013} with $ \bT $ given in Lemma~\ref{theorem:ULDEMMSE}. \endproof
	Using the combiner in \eqref{OptimalMMSE2}, the SINR constraint in \eqref{SINRconstraint} is fulfilled with equality by choosing the optimal power allocation according to the following proposition. In other words, the following proposition provides a necessary 	and sufficient condition for optimality of $ \mathcal{P}1 $.
	\begin{proposition}\label{PropositionP}
		Given the  RBM $ \bPhi $ and the decoder matrix $ \bV $, the optimal power vector $ \bp^{\star} $, accounting for HIs, is obtained geometrically fast as the positive solution to the fixed-point equation given by
		\begin{align}
			\!	\frac{1}{\bar{\tau}^{\star}}\bp^{\!\star}&\!=\left(\etav \!\circ\! \bu\!\circ \!\left(\kappa_{\mathrm{UE}}\deltav \!\circ\!\bONE+\bONE\betav^{\T}\right)\right)\!\bp^{\!\star},\label{Prop1Equation}
		\end{align}
		where $ \bar{\tau}^{\star} $ is the deterministic optimal weighted SINR, $\deltav=\left[\delta_{1}, \ldots,\delta_{K}\right]^{\T} $, and $ \bu=\left[\frac{1}{G_{11}}, \ldots,\frac{1}{G_{KK}}\right]^{\T} $.
	\end{proposition}
	\proof 
	%	Refer to Appendix~\ref{PropositionPproof}.
	We define $\deltav=\left[\delta_{1}, \ldots,\delta_{K}\right]^{\T} $ and $ \bu=\left[\frac{1}{\delta_{1}}, \ldots,\frac{1}{\delta_{K}}\right]^{\T} $. Now, the weighted SINR can be written as
	\begin{align}
		\frac{\bar{\tau}_{k}}{\eta_{k}}\!=\!\frac{p_{k}}{\left(\etav \!\circ\! \bu\!\circ \!\left(\kappa_{\mathrm{UE}}\deltav \!\circ\!\bp+\!\bONE\right)\right)_{k}},\nn
	\end{align}
	where $ \bar{\tau}_{k} $ is the deterministic weighted SINR. Finally, taking advantage of \eqref{Maximization3}   and that, at optimality, the weighted
	SINR for different UEs is the same, we result in \eqref{Prop1Equation}, which converges geometrically fast as follows from 	the remark after Theorem 1 in \cite{Krause2001}.
	\endproof

	Given the decoder matrix $ \bV $ and the optimal power vector $ \bp^{\star} $, the design of the IRS RBM $ \bPhi $ is obtained by means of the optimization problem
	\begin{align}\begin{split}
			(\mathcal{P}2)~~~~~~~\max_{\bPhi} ~~~	&\bar{\tau}^{*}\\
			\mathrm{s.t}~~~&|\phi_{n}|=1,~~ n=1,\dots,N,
		\end{split}\label{Maximization} 
	\end{align}
	which is a maximization problem with a unit-modulus constraint regarding $ \phi_{n} $ that can be solved by using projected gradient
	ascent until converging to a
	stationary point as in \cite{Kammoun2020}. 
	% Basically, $ (\mathcal{P}2) $ concerns the maximization with respect to $ \bar{\tau} $ since $\eta_{k} $ is independent of $ \phi_{n} $ . 
	
	Let $ \bs^{i} =[\phi_{1}^{i}, \ldots, \phi_{N}^{i}]^{\T}$ be the
	the induced phases at step $ i $ and $ \bq^{i} $ be the adopted ascent direction at step $ i $ with $ [\bq^{i}]_{n}= \pdv{\bar{\tau}^{*}}{\phi_{n}^{*}} $(given by Lemma \eqref{lemmaDerivative}),
	the next iteration point is given by
	\begin{align}
		\tilde{\bs}^{i+1}&=\bs^{i}+\mu \bq^{i}\\
		\bs^{i+1}&=\exp\left(j \arg \left(\tilde{\bs}^{i+1}\right)\right),
	\end{align}
	where $ \mu $ is the step size, computed at each iteration by means of the backtracking line search \cite{Boyd2004}.
	The projection problem $ \min_{|s_{n}|=1, n=1,\ldots,N}\|\bs-\tilde{\bs}\|^{2} $ provides the solution while satisfying the unit-modulus constraint. 
	
	%non-convex because of the non-convexity of the objective function $ \bar{\gamma}_{k} $, and the unit-modulus constraint regarding $ \phi_{n} $.
	\begin{lemma}\label{lemmaDerivative}
		The derivative of $ 	\bar{\tau}^{\star} $ with respect to $ \phi_{n}^{*} $ is given by the fixed-point equation \eqref{deriv4} at the top of the next page.
	\end{lemma}
	\proof Refer to Appendix~\ref{Lemma1}.\endproof

After establishing the optimal receiver, power allocation, and RBM, we  combine them by using alternate optimization to find a locally optimal solution. Note that the non-convexity of $ \left(\mathcal{P}1\right) $ cannot guarantee any global optimality. 
		
	Since each subproblem achieves an optimal solution, the objective function  of $ \left(\mathcal{P}1\right) $ is non-decreasing over iterations. Moreover, the optimal value of the objective function is bounded from above due to the power constraint. Hence, the proposed algorithm converges.
		
	The proposed algorithm is quite advantageous compared to algorithms based on instantaneous CSI  since the power allocation and the phase shifts converge to deterministic values that depend on statistical CSI. Hence, they can be a priori calculated and stored while can be updated at  every several coherence intervals due to variation of these channel statistics. On the contrary, instantaneous CSI algorithms would require frequent optimization taking place at each coherence interval.
		\begin{figure*}
		\begin{align}
			\pdv{\bar{\tau}^{\star}}{\phi_{n}^{*}} \! &=\!\frac{p_{k}}{\eta_{k}M\left(1+\kappa_{\mathrm{UE}}\delta_{k}\right)^{2}}\al\bigg[\bH_{1}^{\H}\Big(\beta_{2,k}\bT\bH_{1} \bPhi\tilde{\bR}_{\mathrm{IRS},k}-\kappa_{\mathrm{BS}}\sum_{i=1 }^{K}\!\frac{p_{i}\beta_{2,i}}{M}(\Id_{M}\circ \bT\bR_{k}\bT)\bH_{1} \bPhi\tilde{\bR}_{\mathrm{IRS},i}\!\nn\\
			&-		
			\left(1\!+\! \kappa_{\mathrm{UE}}\right)\sum_{i\ne k}^{K}\!\frac{p_{i}\beta_{2,i}}{M\left(1+\delta_{i}\right)} \bT\bR_{k}\bT\bH_{1} \bPhi\tilde{\bR}_{\mathrm{IRS},i}					\Big)\bigg]_{n,n}-
			\left(1\!+\! \kappa_{\mathrm{UE}}\right)\sum_{i\ne k}^{K}\!\frac{p_{i}\beta_{2,i}\tr(\bT \bR_{k}\bT \bR_{i})}{M\left(1+\delta_{i}\right)^{2}}\pdv{\delta_{i}}{\phi^{*}_{n}}.\label{deriv4}
		\end{align}
		\hrulefill
	\end{figure*}

	%Having obtained the optimal HIs-aware MMSE receiver, the SINR constraint in \eqref{SINRconstraint} is fulfilled with equality by choosing
	%\begin{align}
	%\bp=\arg \max_{ \bw^{\T}\bp = p_{\mathrm{max}}} \min_{k} &\frac{\gamma_{k}\left(\bV, \bP, \bPhi\right)}{\eta_{k}}.
	%\end{align}

	\section{Numerical Results}\label{Numerical} 
	We consider a uniform linear array (ULA) and a uniform planar array (UPA) for the configuration of the BS and IRS, respectively. In particular, we have $ d_{\mathrm{BS}} = d_{\mathrm{IRS}}=0.5\lambda $ while $ \theta_{1,n} $, $ \psi_{1,n} $ are uniformly distributed between $ 0 $ to $ \pi $ and $ 0 $ to $ 2\pi $, respectively. Also, $ \theta_{2,n}= \pi- \theta_{1,n} $, $ \psi_{2,n}=\pi+ \psi_{1,n}$. Moreover, we employ the 3GPP Urban Micro (UMi) scenario from TR36.814 for a carrier frequency of $ 2.5 $ GHz and noise level $ -80 $ dBm, where the path losses for $ \bh_{2,k} $ and $ \bH_{1} $ are generated based on the NLOS and LOS versions, respectively \cite{Kammoun2020}. Specifically, therein, the overall path loss for the IRS-assisted
	link is $ 	\beta_{k}=	\beta_{1,k}	\beta_{2,k} $, where
	\begin{align}
		\beta_{1,k}=C_{1} d_{\mathrm{BS}-\mathrm{IRS}}^{-\nu_{1}},~	~\beta_{2,k}=C_{2} d_{\mathrm{IRS}-\mathrm{UE}_{k}}^{-nu_{2}}
	\end{align}
	with $ C_{1}=26 $ dB, $ C_{2}=28 $ dB, $ \nu_{1} =2.2$, $ \nu_{2} =3.67$. The variables $ d_{\mathrm{BS}-\mathrm{IRS}} $ and $ d_{\mathrm{IRS}-\mathrm{UE}_{k}} $ express the distances between the BS and IRS, and the IRS and UE $ k $, respectively.
	The penetration losses of the IRS-assisted links are assumed negligible by deploying the IRS higher than the BS. For $ \beta_{\mathrm{d},k} $, we assume the same parameters as for $ \beta_{2,k} $, but we also consider an additional penetration loss equal to $ 15~\mathrm{dB} $. We use $ 5 $ dBi antennas at the BS and IRS, and $\bR_{\mathrm{BS},k}, \bR_{\mathrm{IRS},k} $ are generated as in \cite{Bjoernson2020,Kammoun2020}, respectively. The size of each IRS-element dimension is $ \lambda/4 $. The ``solid'' and ``dashed'' lines correspond to $ p_{\mathrm{max}}=0~\mathrm{dB} $ and $ p_{\mathrm{max}}=20~\mathrm{dB} $, respectively, while different line symbols correspond to different values of impairments, given by $ \kappa_{\mathrm{BS}}=\kappa_{\mathrm{UE}}=\{0, 0.05^2, 0.1^{2}\} $, respectively. The optimization takes place by choosing arbitrary values for the phase shifts and the power as in \cite{Cai2011}.  For simplicity, we assume $ \al=1 $ and that all data streams have the same priority $ (\etav=\bONE) $ and power weight
	associated ($ \betav=\frac{1}{K}\bONE $).
	%	\begin{figure}[!h]
	%		\begin{center}
	%			\includegraphics[width=0.9\linewidth]{Graph2.pdf}
	%			\caption{\footnotesize{Asymptotic weighted max-min user rate versus the number of IRS elements $N$ of an IRS-assisted mMIMO system ($ M=80 $, $ K=5 $) for varying RATHIs $ \kappa_{\mathrm{BS}} $, $\kappa_{\mathrm{UE}}$ and transmit power budget $ p_{\mathrm{max}} $. }}
	%			\label{Fig1}
	%		\end{center}
	%	\end{figure}
	%	\begin{figure}[!h]
	%		\begin{center}
	%			\includegraphics[width=0.9\linewidth]{Graph1.pdf}
	%			\caption{\footnotesize{Asymptotic weighted max-min user rate versus the number of BS antennas $M$ of an IRS-assisted mMIMO system ($ N=80 $, $ K=5 $) for varying RATHIs $ \kappa_{\mathrm{BS}} $, $\kappa_{\mathrm{UE}}$ and transmit power budget $ p_{\mathrm{max}} $.}}
	%			\label{Fig2}
	%		\end{center}
	%	\end{figure}
	
	Fig. 1.(a) depicts the minimum uplink user rate $ \log_{2}\left(1+\bar{\tau}\right) $ with respect to the number of IRS elements for different IRS-HIs and correlation conditions with no AT-HIs. The line describing the ideal case (no HIs) appears for comparison. In the case of uniform phase noise $ m=0 $, the rate takes the lowest value since $ \bR_k $ does not depend on the phase shifts. However, when the Von Mises PDF is assumed (``solid'' lines), the IRS can be optimized and the rate increases. In particular, based on the concentration parameter $ \kappa_{\tilde{\theta}} $, we observe that its decrease results in the decrease of the rate. In the special case $ \kappa_{\tilde{\theta}} =0$ (``star'' symbols), the line coincides with the line describing the uniform distribution. Moreover, if no IRS correlation is assumed (``$ \times $'' symbols), the rate is lower since the rate cannot be maximized due to IRS exploitation (see Rem. 1). Also, we have added a ``dotted'' line corresponding to $ \kappa_{\tilde{\theta}}=2 $ that describes a ``naive'' scheme, where the impact of no HIs is taken into account during the optimization design. Its lower SINR reveals the robustness of the proposed design.
	
	In Fig. 1.(b), we illustrate the minimum uplink user rate with respect to the number of IRS elements for different SNR values and varying AT-HIs (no IRS-HIs). Despite that the rate generally increases with SNR, we now observe its increase with $ N $. In the case of perfect hardware, the rate appears no ceiling as $N $ increases, but it saturates in practice, where AT-HIs are met. Moreover, the saturation appears earlier in the case of higher SNR ($ 20~\mathrm{dB} $). Although, the more severe AT-HIs result in higher degradation, we depict the impact of the BS additive distortion $ \kappa_{\mathrm{BS}} $ as $ N $ increases by  the ``dot'' lines in the circle. We notice that the curves converge to the same value when $ N \to \infty $, which means that the impact of $ \kappa_{\mathrm{BS}} $ is negligible at large $ N $, i.e., the larger the IRS, the more beneficial the communication with mMIMO despite the use of low-quality transceiver hardware. Furthermore, the ``star'' lines correspond to Monte Carlo simulations verifying the DE results.
	
	Fig. 1.(c) shows the minimum user rate versus the number of BS antennas for varying AT-HIs (no IRS-HIs). Especially, we observe that the variations of $ M $ and $ N $ present similar behavior. In the case of perfect hardware, the achievable rate increases unboundedly as $M \to \infty $. However, when the AT-HIs are taken into account, we observe finite limits. Notably, the lower quality (more severe AT-HIs) results in larger degradation. Moreover, apart from that, a higher SNR leads to a higher rate, we observe that the convergence to saturation between different SNRs is different since AT-HIs are power-dependent. Thus, at $ 20~\mathrm{dB} $, the rate saturates faster, i.e., the majority of the multi-antenna gain takes place at low $ M $, but still, a large $M $ contributes to larger multiplexing and inter-user interference mitigation. Hence, an IRS-assisted mMIMO works better at high SNR values. In addition, we observe that, at $ 0~\mathrm{dB} $, the convergence requires more antennas.
	%	, which means that the impact of AT-HIs as $ M $ increases.
	\begin{figure*}[t]
		\begin{minipage}{0.33\textwidth}
			\centering
			\includegraphics[trim=0cm -0.20cm 0cm 0.2cm, clip=true, width=2.2in]{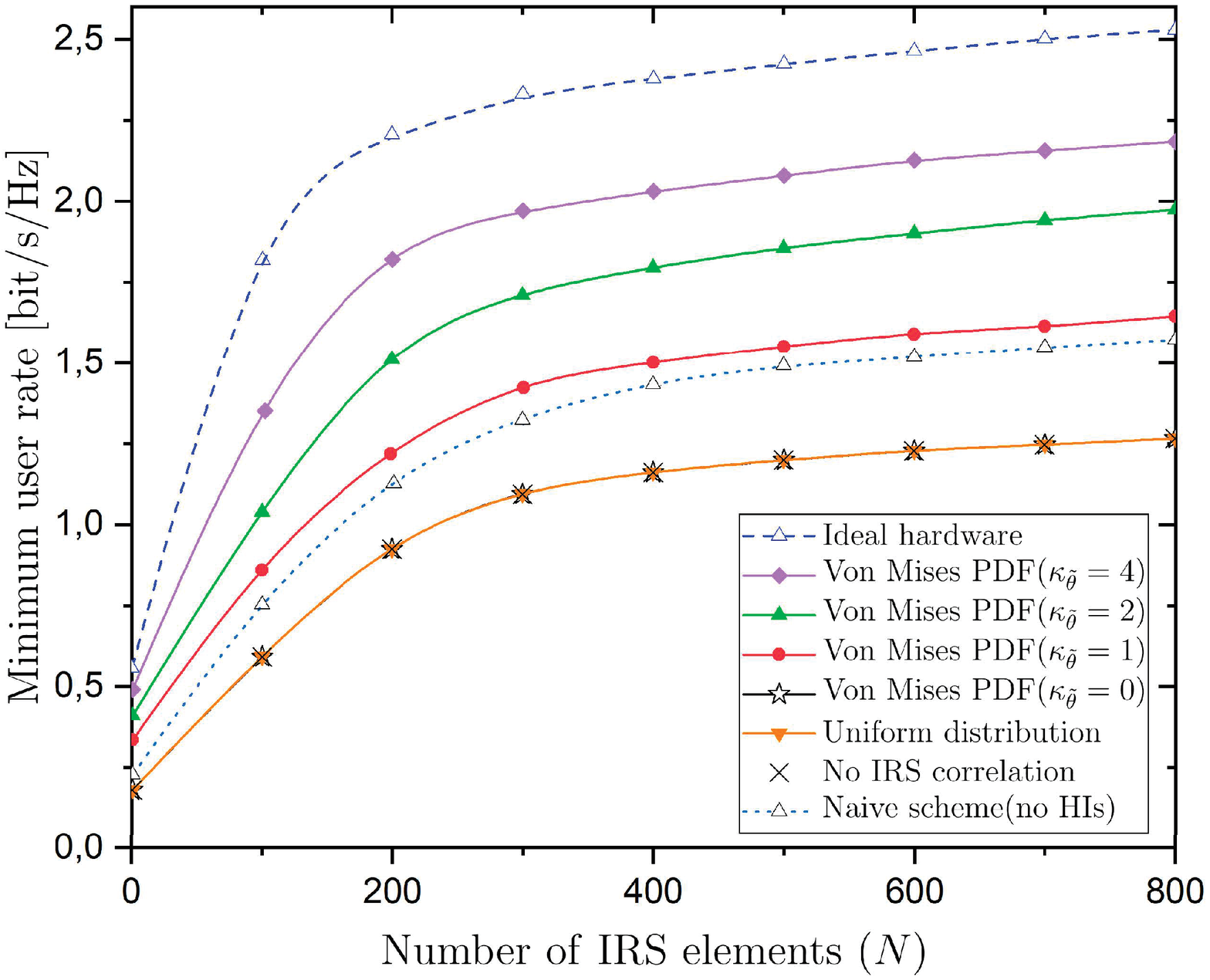} \vspace*{-0.2cm}
			\\ $(a)$
			\label{FigCorrvsUncorrIRS}
			\vspace*{-0.2cm}
		\end{minipage}
		\begin{minipage}{0.33\textwidth}
			\centering
			\includegraphics[trim=0cm -0.20cm 0cm 0.2cm, clip=true, width=2.2in]{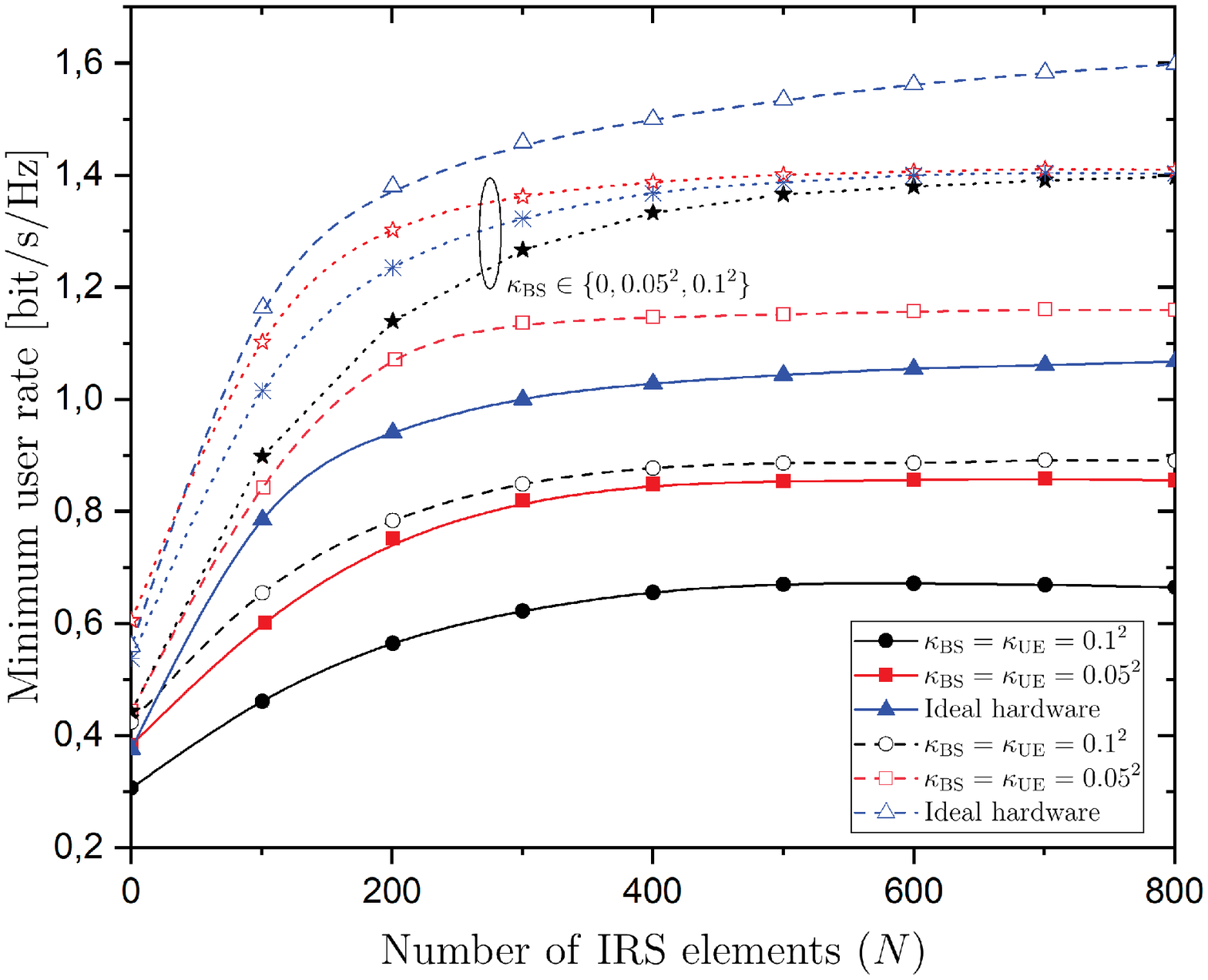} \vspace*{-0.2cm}
			\\$(b)$
			%\caption{The coverage probability versus the different required information rate [b/s/Hz] with the indirect channel only.}
			\label{FigCorrvsUncorrN}
			\vspace*{-0.2cm}
		\end{minipage}
		%\end{figure}
		%\begin{figure}[t]
		\begin{minipage}{0.33\textwidth}
			\centering
			\includegraphics[trim=0cm -0.20cm 0cm 0.2cm, clip=true, width=2.2in]{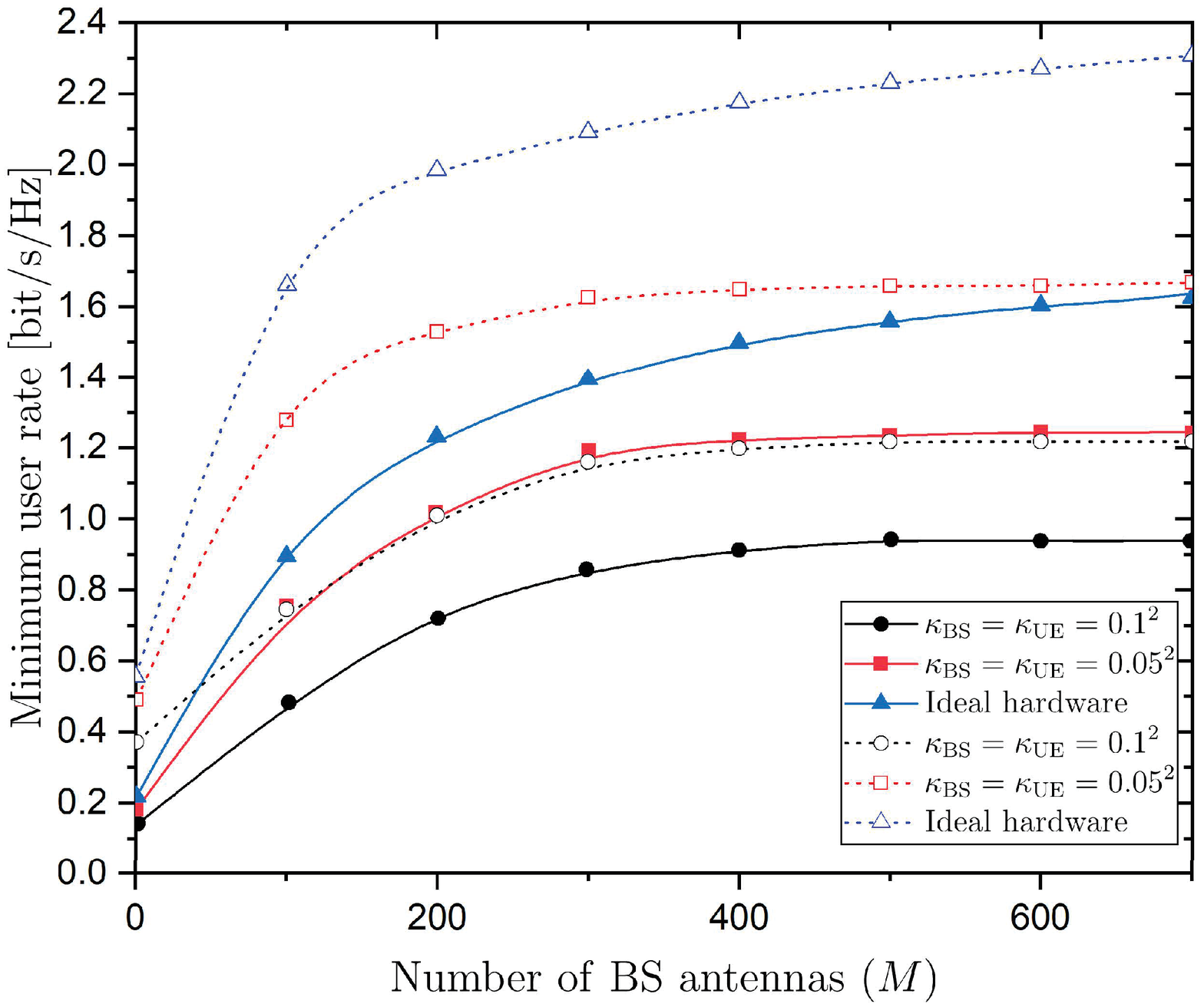}\vspace*{-0.2cm}\\$(c)$
			%\caption{The coverage probability versus the different required information rate [b/s/Hz] with either spatial correlated or uncorrelated Rayleigh fading channels.}
			\label{FigCorrvsUncorrM}
			\vspace*{-0.2cm}
		\end{minipage}
		\caption{Asymptotic weighted max-min user rate of an IRS-assisted mMIMO system for varying $(a)$ IRS-HIs and correlation (correlated/uncorrelated Rayleigh fading channels) conditions versus the number of IRS elements $N$ ($ M=80 $, $ K=5 $); AT-HIs in terms of $ \kappa_{\mathrm{BS}} $, $\kappa_{\mathrm{UE}}$ and transmit power budget $ p_{\mathrm{max}} $ versus $(b)$ the number of IRS elements $N$ ($ M=80 $, $ K=5 $); and $(c)$ the number of BS antennas $M$ ($ N=80 $, $ K=5 $).}
		%			\label{Fig2}
		\vspace{-0.65cm}
	\end{figure*}
% \begin{figure}[!h]
%	\begin{center}
%		\includegraphics[width=0.6\linewidth]{Graph3.pdf}
%		\caption{\footnotesize{\textcolor{blue}{Asymptotic weighted max-min user rate of an IRS-assisted mMIMO system versus $30 $ channel realizations  of an IRS-assisted MIMO system ( $N=80$, $ M=80 $, $ K=5 $, $ \kappa_{\tilde{\theta}}=2 $, $ \kappa_{\mathrm{BS}}=\kappa_{\mathrm{UE}}= 0.1^{2} $, $ p_{\mathrm{max}}=20~\mathrm{dB} $).} }}
%		\label{Fig8}
%	\end{center}
%\textcolor{blue}{The non-convexity of the optimization problem suggests that its solution depends on the initial point, i.e., different initial points result in different locally optimal solutions. Fig. \ref{fig9}.(a) shows the impact of the selection of the initialization on the asymptotic rate by accounting for $ 30 $ channel realizations. The initialization of Alg. 1 assumes that $ \bs^{0} =\exp\left(j\pi/2\right)\one_{N}$ as mentioned in its description. Alg. 1-Test assumes the best initial point out of $ 100 $ random initial points for each channel instance. The figure shows that different initializations result in different solutions and that the sum SE in both cases is almost the same, which means that this  phase shifts selection for initialization is a good choice. }

%\end{figure}
	\section{Conclusion} \label{Conclusion} 
	This paper provided a thorough investigation of the impact of AT-HIs and IRS-HIs on an IRS-assisted mMIMO system. Specifically, we obtained the optimal asymptotic max-min weighted uplink SINR by	optimizing the transmit power, the HIs-aware receiver, and the RBM. Remarkably, this asymptotic SINR, being dependent only on the large-scale statistics, allows optimizing the transmit power and the RBM only on every several coherence when these statistics change. We verified the tightness of the analytical expression by  simulations even for practical system dimensions. Moreover, we shed light on the impact of HIs on the asymptotic max-min weighted SINR by varying the hardware quality at both the IRS and the transceiver.

	\begin{appendices}
				\section{Proof of Lemma~\ref{lemmaDerivative}}\label{Lemma1}
		After adding and subtracting $ 1 $ in the numerator of \eqref{SINRDE1}, the partial derivative of $\bar{\tau}^{\star} $ is written as
		\begin{align}
			\pdv{\bar{\tau}^{\star}}{\phi_{n}^{*}} =\frac{p_{k}}{\eta_{k}}\frac{1}{\left(1+\kappa_{\mathrm{UE}}\delta_{k}\right)^{2}}\pdv{\delta_{k}}{\phi_{n}^{*}}.\label{deriv3}
		\end{align}
		We have
		\begin{align}
			&\pdv{\delta_{k}}{\phi^{*}_{n}}=\frac{1}{M}\tr\left(\pdv{\bR_{k}}{\phi^{*}_{n}}\bT+\bR_{k}\pdv{\bT}{\phi^{*}_{n}}\right)\label{der1}\\
			&=\frac{1}{M}\tr\bigg(\pdv{\bR_{k}}{\phi^{*}_{n}}\bT\bigg)-\frac{1}{M}\tr\bigg(\bR_{k}\bT\pdv{ \bT^{-1}}{\phi^{*}_{n}} \bT \bigg)\label{der2}\\
			&\!=\!\frac{1}{M}\tr\!\bigg(\!\pdv{\bR_{k}}{\phi^{*}_{n}}\bT\bigg)\!-\!\frac{1}{M}\kappa_{\mathrm{BS}}\tr\!\bigg(\!\bR_{k}\bT\bigg(\!\sum_{i=1 }^{K}\!\frac{p_{i}}{M}\!\pdv{\big( \Id_{M}\!\circ\!\bR_{i}\big)}{\phi^{*}_{n}}\! \bT \!\bigg)\nn\\
			&-\!\frac{1}{M}\!\tr\!\bigg(\!\bR_{k}\bT\bigg(\!\sum_{i\ne k}^{K}\!\frac{p_{i}\left(1\!+\! \kappa_{\mathrm{UE}}\right)}{M\left(1+\delta_{i}\right)}\big(\!\pdv{ \bR_{i}}{\phi^{*}_{n}} \!-\!\frac{\pdv{\delta_{i}}{\phi^{*}_{n}}}{\left(1\!+\!\delta_{i}\right)}\bR_{i}\!\big) \!\bT \!\bigg)\!
			,\label{del1}
		\end{align}\\
	where, 
%	in \eqref{der1}, we have considered the linearity of the $ \tr $ operator, and 
	in \eqref{der2}, we  used 
%	\cite[Eq. 40]{Petersen2012} providing 
the derivative of the inverse matrix $ \bT $.
		%		\begin{lemma}\label{traceProd}
		%			Let 	 $ \bA \in\mathbb{C}^{M\times M} $ be independent of $ \bPhi$ and $\bR_{k}= \beta_{2,k}\bH_{1} \bPhi$ $\tilde{\bR}_{\mathrm{IRS},k}\bPhi^{\H}\bH_{1}^{\H} $, then
		%			\begin{align}
		%			\tr\left( \!\!\bA\pdv{\bR_{k}}{\phi^{*}_{n}}\!\right) =\al\beta_{2,k}[\bH_{1}^{\H}\bA\bH_{1} \bPhi\tilde{\bR}_{\mathrm{IRS},k}]_{n,n}.
		%			\end{align}\end{lemma}
		%		\proof We have
		%		\begin{align}
		%		\!\tr\!\left(\! \!\bA\pdv{\bR_{k}}{\phi^{*}_{n}}\!\right) &\!=\!\sum_{i,j}[\bA]_{ij}\pdv{[\bR_{k}]_{ji}}{\phi^{*}_{n}}\\
		%		&\!=\!\al\beta_{2,k}\sum_{i,j}[\bA]_{ij}[\bH_{1} \bPhi\tilde{\bR}_{\mathrm{IRS},k}]_{jn}[\bH_{1}^{\H} ]_{in}^{\T}\\
		%		&\!=\!\al\beta_{2,k}[\bH_{1}^{\H}\bA\bH_{1} \bPhi\tilde{\bR}_{\mathrm{IRS},k}]_{nn},
		%		\end{align}
		%		since $ \pdv{[\bR_{k}]_{ji}}{\phi^{*}_{n}}=\al\beta_{2,k} [\bH_{1} \bPhi\tilde{\bR}_{\mathrm{IRS},k}]_{jn}[\bH_{1}^{\H} ]_{in}^{\T} $.
		%		\endproof
		
		Use of \cite[Lem. 1]{Papazafeiropoulos2021}
		%		Lemma \ref{traceProd} 
		in \eqref{del1} and substitution into \eqref{deriv3} concludes the proof after several algebraic manipulations.
	\end{appendices}
	\bibliographystyle{IEEEtran}
	
	\bibliography{mybib}
\end{document}